\documentclass[11pt]{article}
\usepackage{amssymb}
\usepackage{amsthm}
\usepackage{latexsym}
\usepackage[ruled,vlined,algosection]{algorithm2e}
\usepackage{amsmath}
\newtheorem{theo}{Theorem}[section]
\newtheorem{prop}[theo]{Proposition}
\newtheorem{cor}[theo]{Corollary}
\newtheorem{lemma}[theo]{Lemma}
\theoremstyle{definition}
\newtheorem{defi}[theo]{Definition}
\newtheorem{exa}[theo]{Example}
\newtheorem{rem}[theo]{Remark}
\numberwithin{equation}{section}
\numberwithin{algocf}{section}

\newcommand{\F}{{\mathbb F}}

\newcommand{\Fqns}{\mathbb{F}_{q^n}^*}

\newcommand{\cA}{{\mathcal A}}
\newcommand{\cC}{{\mathcal C}}

\newcommand{\cM}{{\mathcal M}}

\newcommand{\cU}{{\mathcal U}}
\newcommand{\cV}{{\mathcal V}}
\newcommand{\cW}{{\mathcal W}}

\newcommand{\rank}{\mbox{${\rm rk}$}}
\newcommand{\im}{\mbox{\rm im}}
\newcommand{\spann}{\mbox{\rm span}\,}
\renewcommand{\mod}{\mbox{\rm mod}\,}
\newcommand{\mmid}{\mbox{$\,|\,$}}

\newcommand{\GL}{\mathrm{GL}}

\newcommand{\ds}{\textup{d}_{\rm{S}}}
\newcommand{\dr}{\textup{d}_{\rm{R}}}

\newcommand{\dist}[1]{\textup{d}_s(#1)}
\newcounter{alp}
\newcounter{ara}
\newcounter{rom}

\newenvironment{alphalist}{\begin{list}{(\alph{alp})\hfill}{\usecounter{alp}
     \topsep0ex \labelwidth.6cm \leftmargin.6cm \labelsep0cm
     \rightmargin0cm \parsep0ex \itemsep0ex
     \partopsep0ex}}{\end{list}}

\begin{document}
%%%%%%%%%%%%%%%%%%%%%%%%%%%%%%%%
\title{Construction of Subspace Codes through Linkage}
\date{\today}
\author{Heide Gluesing-Luerssen$^*$
and Carolyn Troha\footnote{HGL was partially supported by the National
Science Foundation Grant DMS-1210061. HGL and CT are with the Department of Mathematics, University of Kentucky, Lexington KY 40506-0027, USA;
\{heide.gl, carolyn.troha\}@uky.edu.}}

\maketitle

{\bf Abstract:}
A construction is presented that allows to produce subspace codes of long length using subspace codes of shorter length
in combination with a rank metric code.
The subspace distance of the resulting code, called linkage code, is as good as the minimum subspace distance of the constituent codes.
As a special application, the construction of the best known partial spreads is reproduced.
Finally, for a special case of linkage, a decoding algorithm is presented which amounts to decoding
with respect to the smaller constituent codes and which can be parallelized.



\section{Introduction}\label{S-Intro}
In~\cite{KoKsch08} Koetter and Kschischang developed an approach to random network coding where
the encoded information is represented as subspaces of a given ambient space.
This accounts for the unknown network topology by assuming that any linear combination of packets may
occur at the nodes of the network.

This approach has led to the area of subspace codes and specifically to intensive research efforts on constructions
of subspace codes with large subspace distance
\cite{KoKu08,KSK09,SiKsch09,EtSi09,GaBo09,EKW10,GPB10,EtVa11,GaPi13,EtSi13,RoTr13,TMBR13,GLMT15}.
Most of the research focuses on constant-dimension codes (CDC's), that is, codes where all subspaces have the same dimension.

One direction for constructing CDC's is based on so-called cyclic orbit codes~\cite{KoKu08,EKW10,RoTr13,TMBR13,GLMT15}, which
are orbits of a subspace in the $\F_q$-vector space~$\F_{q^n}$ under the natural action of $\F_q^{\ast}$.
While the resulting codes have very beneficial algebraic structure, they do not have large cardinality in general.
Taking unions of such codes leads to cyclic subspace codes which still have nice structure, but it remains an open problem how to
take unions of cyclic orbit codes without decreasing the distance.

A second major research direction is based on rank-metric codes as introduced and studied earlier
by Delsarte~\cite{Del78} and Gabidulin~\cite{Gab85}.
Lifting rank-metric codes~\cite{SKK08} is a very simple construction which results in subspace codes where the
reduced row echelon form of each subspace has its identity matrix in the leftmost position.
While these codes are asymptotically good~\cite{KSK09}, they can still be improved upon.
Through a careful study of general reduced row echelon forms, Etzion and Silberstein~\cite{EtSi09} could enlarge lifted
rank-metric codes.
This is known as the echelon-Ferrers construction (or multilevel construction) and has led to the best
codes known at that time.
In the same paper it is also demonstrated how to decode these codes.
In~\cite{EtSi13} the same authors could enlarge their codes even further by making use of the pending dot
construction initiated by Trautmann/Rosenthal~\cite{TrRo10}.
A different way of enlarging lifted rank-metric codes has been presented by Skachek in~\cite{Ska10}, where for
a specific case of the resulting codes also a decoding algorithm is developed.

A slightly different approach is taken in~\cite[Sec.~III]{EtVa11} by Etzion/Vardy and~\cite{GoRa14} by Gorla/Ra\-vagnani.
Therein, the authors focus on CDC's with the best theoretical distance.
This distance is achieved when any two distinct subspaces of the code intersect trivially.
In~\cite{EtVa11,GoRa14} the authors present constructions of such codes with large cardinality (in fact, the same
cardinality).
In finite projective geometry, collections of pairwise trivially intersecting subspaces are known as partial spreads, or as spreads if the
subspaces also cover the entire ambient space.
The cardinalities obtained in~\cite{EtVa11,GoRa14}
have also been achieved via constructions based on finite geometry~\cite{Beu75,DrFr79}.

In this paper we will present a surprisingly simple construction of subspace codes of long length by linking matrix representations
of subspace codes of shorter length and making use of a rank-metric code.
The distance of the resulting subspace code, called a linkage code, is the minimum of the distances of the involved subspace codes
and the related lifted rank-metric code.

By way of an example we will compare our construction to the echelon-Ferrers construction~\cite{EtSi09}
and its improvement~\cite{EtSi13} as well as to the best codes found by sophisticated computer search.
We will see that our codes improve upon those constructed in~\cite{EtSi09,EtSi13} and for lengths where no computer
search has been conducted, it leads to the best codes known so far.

In Section~\ref{S-ParSpr} we show that specific instances of linkage reproduce the constructions of partial spreads
given in~\cite{EtVa11,GoRa14}, and thus linkage may be considered as a generalization of the latter.
In fact, those constructions, which on first sight appear to be quite different,
can be nicely presented in a unified way with the aid of linkage.
Using an optimal binary partial spread of dimension~$3$ and length~$8$, which was found only recently by El-Zanati et al.~\cite{EJSSS10}, we
can improve upon the constructions in~\cite{EtVa11,GoRa14} for the binary case in dimension~$3$ and when the length has
remainder~$2$ modulo~$3$.
This reproduces results derived in~\cite{EJSSS10}.

Finally, in Section~\ref{S-Decoding} we show that for a specific choice of the seed codes, decoding of linkage codes can be reduced to
decoding of the seed codes.
Making use of Gabidulin codes and lifted Gabidulin codes, for which an efficient decoding algorithm has been derived
by Silva et al.~\cite{SKK08}, we obtain linkage codes that can be decoded by an efficient parallelizable algorithm.

%%%%%%%%%%%%%%%%%%%%%%%%%%%%%%%%%%%%%
\section{A Linkage Construction for Subspace Codes}\label{S-Linkage}
%%%%%%%%%%%%%%%%%%%%%%%%%%%%%%%%%%
Let us first recall some basic facts about subspace codes and rank-metric codes.
Throughout we fix a finite field~$\F=\F_q$.
A \emph{subspace code of length~$n$} is simply a non-empty collection of subspaces in~$\F^n$.
The code is called a \emph{constant-dimension code} if all subspaces have the same dimension.
The \emph{subspace distance} of a subspace code~$\cC$ is defined as
$\ds(\cC):=\min\{\ds(\cV,\cW)\mid \cV,\,\cW\in\cC,\,\cV\neq\cW\}$,
where the distance between two subspaces is
\[
   \ds(\cV,\cW):=\dim\cV+\dim\cW-2\dim(\cV\cap\cW).
\]
It is a metric on the space of all subspaces, see \cite[Lem.~1]{KoKsch08}.
If $\dim\cV=\dim\cW=k$, then $\ds(\cV,\cW)=2\big(\dim(\cV+\cW)-k\big)$.
As a consequence, if~$\cC$ is a constant-dimension code of dimension~$k$ then
\begin{equation}\label{e-dmax}
   \ds(\cC)\leq\min\{2k,\,2(n-k)\}.
\end{equation}
If the code~$\cC$ consist of a single subspace of dimension~$k$, we define $\ds(\cC):=\min\{2k,\,2(n-k)\}$.
A constant-dimension code of length~$n$, dimension~$k$, cardinality~$N$ will be called an $(n,N,k)$-code, and it is
a $(n,N,k,d)$-code if its subspace distance is~$d$.

A $k\times m$ rank-metric code is a non-empty subset of $\F^{k\times m}$,
endowed with the rank metric $\dr(A,B):=\rank(A-B)$ (which is indeed a metric, see~\cite{Gab85}).
The rank distance of a rank-metric code~$\cC$ is defined in the usual way as
$\dr(\cC):=\min\{\rank(A-B)\mid A,\,B\in\cC,\,A\neq B\}$.
If~$\cC$ consists of single matrix, we define $\dr(\cC):=\min\{k,\,m\}$.
It is well known (see~\cite[Thms.~5.4,~6.3]{Del78} and~\cite[p.~2]{Gab85}) that if~$m\geq k$ and~$\cC$ is a rank-metric code
in~$\F_q^{k\times m}$ with rank distance~$d$, then
\begin{equation}\label{e-rkdist}
     |\cC|\leq q^{m(k-d+1)}.
\end{equation}
Moreover, there exist rank-metric codes of distance~$d$ and size~$q^{m(k-d+1)}$, and such codes can even be constructed as
linear subspaces of $\F_q^{k\times m}$.
They are called \emph{MRD codes}.
The best known class of linear MRD codes are the Gabidulin codes, derived by Gabidulin in~\cite{Gab85}.
Just recently, other constructions of MRD codes were found by de la Cruz et al.~\cite{CKWW15} and Hernandez/Sison~\cite{HeSi15}.
However, as opposed to Gabidulin codes so far no decoding algorithm is known for the latter codes.

If~$\cC\subseteq\F^{k\times m}$ is an MRD code, then the subspace code
\[
   \widehat{\cC}=\{\im(I_k\mid A)\mid A\in\cC\}
\]
is called a \emph{lifted MRD code}.
Here the notation $\im(M)$ stands for the row space of the matrix~$M$ and
$I_k$ denotes the $k\times k$-identity matrix.
If $\dr(\cC)=d$, then $\ds(\widehat{\cC})=2d$, see \cite[Prop.~4]{SKK08}, and
therefore~$\widehat{\cC}$ is a $(k+m,\,q^{m(k-d+1)},\,k,\,2d)$-code.

The following specific class of MRD codes will be used in the next section when studying partial spreads.
%%%%%%%%%%%%%%%%%%%%%%%%%
\begin{rem}\label{R-COC}
There is a simple construction of MRD codes in $\F^{k\times m}$ of rank distance~$k$.
Let $W\in\F^{k\times m}$ be any matrix of rank~$k$ and $M\in\GL_{m}(\F)$ be the companion
matrix of a primitive polynomial in $\F[x]$ of degree~$m$.
Define
\[
   \cC:=\{WM^l\mid l=0,\ldots,q^{m}-2\}\cup\{0\}\subseteq\F^{k\times m}.
\]
Since $\F[M]\cong\F_{q^m}$ and $|M|=q^m-1$, the code~$\cC$ is a linear rank-metric code
of size~$q^{m}$ and rank distance~$k$.
Hence~$\cC$ is an MRD code.
In fact, it can be shown that~$\cC$ is a Gabidulin code.
\end{rem}
%%%%%%%%%%%%%%%%%%%%%%%%%%

In order to present our linkage construction we need to work with matrix representations of subspaces.
The following terminology will be convenient.

%%%%%%%%%%%%%%%%%%%%%%%%%%%%%%
\begin{defi}\label{D-SC-Repr}
A set of matrices $\cM\subseteq\F^{k\times n}$ is called \emph{SC-representing} if
$\rank(M)=k$ for all $M\in\cM$ and $\im(M)\neq\im(M')$ for all $M\neq M'$.
The induced constant-dimension code $\{\im(M)\mid M\in\cM\}$ is denoted by $\cC(\cM)$.
\end{defi}
%%%%%%%%%%%%%%%%%%%%%%%%%%%%%%%%%

For example, the rank-metric code~$\cC$ in Remark~\ref{R-COC} forms an SC-representing set of a
cyclic orbit code in the sense of \cite{RoTr13,TMBR13,GLMT15}.
Any set of full row rank matrices in reduced row echelon form is an SC-representing set.
In general, an SC-representing set is simply a subset of orbit representatives of the action of $\GL_k(\F)$ on
$\F^{k\times n}$ via left multiplication.

The linkage construction in the following theorem links subspace codes with the aid of a rank-metric code
and results in a subspace code of longer length without compromising the distance.
It makes use of representing matrices.
The theorem generalizes a former construction in~\cite[Thm.~5.1]{GLMT15}.

%%%%%%%%%%%%%%%%%%%%%%%%%%%%%%%%
\begin{theo}\label{T-Linkage}
For $i=1,2$ let~$\cM_i\subseteq\F^{k\times n_i}$ be SC-representing sets of cardinality~$N_i$.
Thus~$\cC_i=\cC(\cM_i)$ is an $(n_i,\,N_i,\,k)$-code.
Let $\ds(\cC_i)=d_i$.
Furthermore, let~$\cC_R\subseteq\F^{k\times n_2}$ be a linear rank-metric code with rank distance $\dr(\cC_R)=d_R$ and
cardinality $|\cC_R|=N_R$.
Define the subspace code~$\cC$ of length $n:=n_1+n_2$ as
$\cC:=\tilde{\cC}_1\cup\tilde{\cC}_2\cup\tilde{\cC}_3$,
where
\begin{align*}
     &\tilde{\cC}_1=\{\im(U\mmid 0_{k\times n_2})\mid U\in\cM_1\}, \\[.5ex]
     &\tilde{\cC}_2=\{\im(0_{k\times n_1}\!\mmid U)\mid U\in\cM_2\}, \\[.5ex]
     &\tilde{\cC}_3=\{\im(U\mmid M)\mid U\in\cM_1,\,M\in\cC_R\backslash\{0\}\}.
\end{align*}
Then~$\cC$ is a $(n,\,N,\,k,\,d)$-code, where $N=N_2+N_1N_R$ and
$d=\min\{d_1,\,d_2,\,2 d_R\}$.
We write $\cC=\cC_1\ast_{\cC_R}\cC_2$ for the resulting linkage code and call~$\cC$ the code obtained by
linking~$\cC_1$ and~$\cC_2$ through~$\cC_R$.
\end{theo}
%%%%%%%%%%%%%%%%%%%%%%%%%%%%%%%%%%%%
The notation $\cC=\cC_1\ast_{\cC_R}\cC_2$ has to be used with care because the code~$\cC$ depends on the
representing sets~$\cM_1$ and~$\cM_2$ and not only on the codes~$\cC_1$ and~$\cC_2$.
Thus the notation $\cM_1\ast_{\cC_R}\cM_2$ is actually the accurate one, but we prefer the former because the properties of the
linkage that we are interested in are associated to the subspace codes~$\cC_1,\,\cC_2$.
The sets~$\cM_1,\,\cM_2$ are mere technicalities in our context, and the notation $\cC_1\ast_{\cC_R}\cC_2$ will not lead to any confusion.

However, in order to illustrate the dependence on the SC-representing sets for~$\cC_1$ and~$\cC_2$ we will present, after the proof,
an example showing that different choices lead in general to different distance distributions of the linkage code.

\begin{proof}
The cardinality of~$\cC$ is clear because the three sets~$\tilde{\cC}_i$ are pairwise disjoint.
Furthermore, it is obvious that $\ds(\tilde{\cC}_i)=\ds(\cC_i)$ for $i=1,2$.
Moreover, it is clear that each subspace in~$\tilde{\cC}_2$ intersects trivially with each subspace
in~$\tilde{\cC}_1$ and $\tilde{\cC}_3$.
Thus $\ds(\cW_1,\cW_2)=2k$ for all $\cW_1\in\tilde{\cC}_2$ and $\cW_2\in\tilde{\cC}_1\cup\tilde{\cC}_3$.

Next, let~$\cU=\im(U\mmid 0)\in\tilde{\cC}_1$ and~$\cV=\im(U'\mmid M)\in\tilde{\cC}_3$.
Thus~$U,\,U'\in\cM_1$ and $M\in\cC_R\backslash\{0\}$.
Then $\rank(M)\geq d_R$ by linearity of the code~$\cC_R$ and so $\dim\ker M\leq k-d_R$, where
$\ker M=\{x\in\F^k\mid xM=0\}$.
Let now $v\in\cU\cap\cV$.
Then $v=x(U\mmid 0)=y(U'\mmid M)$ for some $x,\,y\in\F^k$ and thus $y\in\ker M$ and
$xU=yU'\in\im(U)\cap\im(U')$.
Since the maps from $\cU\cap\cV$ mapping $x(U\mmid 0)$ to~$xU$ and $y(U'\mmid M)$ to~$y$ are both injective, this
yields
\[
  \dim(\cU\cap\cV)\leq\min\{\dim(\im(U)\cap\im(U')),\,\dim(\ker(M))\}.
\]
Now one concludes $\ds(\cU,\,\cV)\geq\min\{d_1,\,2d_R\}$, as desired.

Lastly, let $\cU=\im(U\mmid M),\,\cV=\im(U'\mmid M')\in\tilde{\cC}_3$, and
let~$\cU\neq\cV$, thus $U\neq U'$ or $M\neq M'$.
Let $v\in\cU\cap\cV$.
Then
$v=x(U\mmid M)=y(U'\mmid M')$ for some $x,\,y\in\F^k$.
Hence
\[
   xU=yU'\in\im(U)\cap\im(U')\ \text{ and }\ xM=yM'\in\im(M)\cap\im(M').
\]
If $U\neq U'$, then $\dim(\cU\cap\cV)\leq\dim(\im(U)\cap\im(U'))$ and
$\ds(\cU,\,\cV)\geq \ds(\im(U),\,\im(U'))\geq d_1$.
If $U=U'$, then $x=y$ because $U$ has full row rank.
Moreover, $M-M'\in\cC_R\backslash\{0\}$ and thus $\dim(\ker(M-M'))\leq k-d_R$.
Now $x\in\ker(M-M')$ along with the injectivity of the map $x(U\mmid M)\longmapsto x$
from $\cU\cap\cV$ to $\ker(M-M')$
shows that $\dim(\cU\cap\cV)\leq\dim(\ker(M-M'))\leq k-d_R$, and thus
$\ds(\cU,\,\cV)\geq 2d_R$.
This concludes the proof.
\end{proof}

The following example shows that different choices of the SC-representing sets for~$\cC_1$ and~$\cC_2$ lead
to different distance distributions of the linkage code.
Since we will not further study the distance distribution of linkage codes, we continue to use the notation $\cC_1\ast_{\cC_R}\cC_2$
for the linkage.
%%%%%%%%%%%%%%%%%%%%%%%%%%%
\begin{exa}\label{E-DistDistr}
Let $(n_1,n_2,k,q)=(4,4,2,2)$ and
\[
 \cM_1=\cM_2=\Big\{\begin{pmatrix}1&0&1&0\\0&1&0&0\end{pmatrix},\begin{pmatrix}1&0&0&0\\0&1&0&0\end{pmatrix}\Big\}
 \text{ and }
 \cM_1'=\Big\{\begin{pmatrix}1&1&1&0\\0&1&0&0\end{pmatrix},\begin{pmatrix}0&1&0&0\\1&1&0&0\end{pmatrix}\Big\}
\]
and
\[
  \cC_R=\Big\{\begin{pmatrix}0&0&0&0\\0&0&0&0\end{pmatrix},\begin{pmatrix}1&0&0&0\\0&1&0&0\end{pmatrix},
              \begin{pmatrix}0&1&0&0\\0&0&1&0\end{pmatrix},\begin{pmatrix}1&1&0&0\\0&1&1&0\end{pmatrix}\Big\}.
\]
Note that $\cC(\cM_1)=\cC(\cM_1')$.
We find that in the linkage code $\cM_1\ast_{\cC_R}\cM_2$ (see the notation of the paragraph after Theorem~\ref{T-Linkage})
there exist~$5$ pairs of distinct subspaces with subspace distance~$2$ and all other pairs have subspace distance~$4$, whereas in
the linkage code $\cM_1'\ast_{\cC_R}\cM_2$ only~$3$ pairs have subspace distance~$2$ and all others have subspace distance~$4$.
\end{exa}
%%%%%%%%%%%%%%%%%%%%%%%%%%%%

The next two examples illustrate that we can easily construct very large codes of long length by suitable linkage.

%%%%%%%%%%%%%%%%%%%%%%%%%%%%%%%%%
\begin{exa}\label{E-LargeLink}
We aim at constructing a constant-dimension code over~$\F_2$ of length~$13$,
dimension~$k=3$, and distance~$4$.
Let $n_1=7$ and $n_2=6$.
The largest known codes of dimension~$3$ and length~$7$ (resp.~$6$) with
distance~$4$ have cardinality~$329$ (resp.~$77$), see~\cite[Tables~I and~II]{BrRei14} by Braun/Reichelt as well as \cite{HKK14}
by Honold et al., where it is shown that~$77$ is actually the largest possible size for length~$6$.
We choose these codes for~$\cC_1$ and~$\cC_2$, respectively, and an MRD code~$\cC_R$ in $\F_2^{3\times 6}$
with rank distance~$d_R=2$ and thus cardinality~$N_R=2^{6(3-2+1)}=2^{12}$ due to~\eqref{e-rkdist}.
The resulting linkage code has therefore cardinality $77+2^{12}\cdot329=1,347,661$.
This is lower than the cardinality of the best known code with the same parameters,
which is $1,597,245$, see~\cite[App.]{BrRei14}.
But the latter has been found by extended computer search, whereas the linkage code is readily available
once~$\cC_1$ and~$\cC_2$ have been found.
The linkage code beats the codes that have been found  by Etzion/Silberstein with the aid of the multilevel construction~\cite{EtSi09}
where the best such code has cardinality~$1,192,587$, see~\cite{KoKu08}.
It also beats the modified multilevel construction~\cite{EtSi13} by Etzion/Silberstein, where the best such code has
size $1,221,296$.
This particular construction is a refinement of the pending dot construction appearing first in~\cite{TrRo10} by Trautmann/Rosenthal.
The following table presents the cardinality of further linkage constructions for $q=2,\,k=3,\,\ds=4$ and various lengths~$n$.
We make use of the best codes found in~\cite[Table~II]{BrRei14} for lengths~$6,\ldots,9$ and an MRD code~$\cC_R\subseteq\F^{3\times n_2}$
with rank distance~$2$. Thus~$N_R=2^{2n_2}$ thanks to~\eqref{e-rkdist}.
We also show the largest size obtained via the modified multilevel (MML)
construction~\cite[Thm.~17]{EtSi13} (which always beats the multilevel construction in~\cite{EtSi09}) as well as the largest size known so far.
For $n\leq14$ the latter has been found by computer search~\cite[Tables~I and~II]{BrRei14}, while for $n=15$ no such search has been conducted yet
and linkage with $n_1=9,\,n_2=6$ results in the largest known code.
Note that for all $n$ shown in the table, every partition into $n=n_1+n_2$ leads to a linkage code that is larger than the MML construction.
This is probably due to the fact that the MML construction leads to subspace codes that contain a lifted MRD code.
This restriction also restricts the size of these codes.
\[
\begin{array}{||c|c|c|r|r|r|r|r|r||}
\hline
n&n_1&n_2&N_1\ &N_2\ &\text{Linkage}\ & \text{MML}\quad &\text{Largest Known}\\\hline\hline
12&6&6&77&77&315,469& 305,324 &385,515\\ \hline
13&6&7&77&329&1,261,897& 1,221,296 & 1,597,245\\ \hline
13&7&6&329&77&1,347,661& 1,221,296 &1,597,245\\ \hline
14&7&7&329&329&5,390,665& 4,885,184 & 5,996,178 \\ \hline
14&6&8&77&1,312&5,047,584& 4,885,184 & 5,996,178 \\ \hline
14&8&6&1,312&77&5,374,029& 4,885,184 & 5,996,178 \\ \hline
15&8&7&1,312&329  &21,496,137&19,540,736 & 23,322,701 \\ \hline
15&7&8& 329  &1,312&21,562,656&19,540,736  & 23,322,701 \\ \hline
15&6&9& 77  &5,694&20,190,782&19,540,736& 23,322,701 \\ \hline
15&9&6& 5,694& 77&23,322,701&19,540,736& 23,322,701 \\ \hline\hline
\end{array}
\]
It is worth pointing out that for length~$n=13$ the largest known cardinality~$1,597,245$ is actually the optimum by the anticode bound,
and the existence of a code
with that size has been established by Braun et al.\ in~\cite{BEOVW13} via a 2-analogue of a Steiner triple system.
\end{exa}

%%%%%%%%%%%%%%%%%%%%%%%%%%%%
\begin{exa}\label{E-MRDLinkage}
Let us consider Theorem~\ref{T-Linkage} for the case where~$\cC_1$ is a lifted MRD code and~$\cC_R$ is an MRD
code.
For~$\cC_2$ we may choose a lifted MRD code or an arbitrary subspace code.
Let us consider the case where~$\cC_2$ is a lifted MRD code.
Thus, let $n=n_1+n_2$, where $n_i\geq 2k$ for $i=1,2$, and~$\cC_i$ be the lifting of an MRD code in~$\F^{k\times(n_i-k)}$
with rank distance~$d$.
Then $|\cC_i|=q^{(n_i-k)(k-d+1)}$ and $\ds(\cC_i)=2d$.
Moreover, let $\cC_R\subseteq\F^{k\times n_2}$ be an MRD code of rank distance~$d$.
Thus, $|\cC_R|=q^{n_2(k-d+1)}$.
By Theorem~\ref{T-Linkage} the linkage code $\cC_1\ast_{\cC_R}\cC_2$ has subspace distance~$2d$ and cardinality
$q^{(n_2-k)(k-d+1)}+q^{(n-k)(k-d+1)}$.
Note that the second term is the cardinality of a lifted MRD code in~$\F^n$ with subspace distance~$2d$.
Thus, linkage always results in a better code than lifting.
In fact, with our choice the code $\tilde{\cC}_1\cup\tilde{\cC}_3$ in Theorem~\ref{T-Linkage} is a
lifted MRD code and thus the cardinality of the linkage code is clearly larger than that of a lifted MRD code.
Furthermore we observe that only the first term depends on the partition $n=n_1+n_2$, and that
the cardinality of $\cC_1\ast_{\cC_R}\cC_2$ is largest when~$n_2$ is largest.
The following table shows the size of the linkage construction for $q=2,\,k=3,\,\ds=4$ and various lengths.
In each case,~$\cC_1$ is a
lifted MRD code of distance~$4$ and~$\cC_R$ is an MRD code of rank distance~$2$.
Each given length~$n$ is split into $n=n_1+n_2$ such that $n_2$ is maximal subject to  $n_i\geq 2k$ for $i=1,2$.
In the column denoted by ``$\text{Link}_{\text{largest}}$'' we present the cardinality of the linkage code where we use
the largest known subspace code for~$\cC_2$.
In the column ``$\text{Link}_{\text{MRD}}$'' we use a lifted MRD code for~$\cC_2$.
For comparison we also show the size,~$2^{2(n-3)}$, of a lifted MRD code of length~$n$.
It should be noted that the linkage codes are smaller than the codes obtained from the MML construction; see the previous table.
This is explained by the fact that the MML construction is a careful design to create additional subspaces without
compromising the distance.
It may be regarded as a replacement of the code $\tilde{\cC}_2$ in Theorem~\ref{T-Linkage} by a larger set,
where the zero block matrix is replaced by suitable matrices.
In the column ``$\text{Extended Lifted MRD}$'' we illustrate that our codes are slightly smaller than those
constructed in~\cite{Ska10} by Skachek\footnote{One can show that in our situation
the optimal design choice for the parameter $h_{n}=h_{\ell+m}$ in~\cite{Ska10} is $h_{n}=0$, and therefore the cardinality of the resulting
code is as in Section~IV.C of~\cite{Ska10}.}, which are also subspace codes containing a lifted MRD code (and are
smaller than the MML codes).
\[
\begin{array}{||c|c|c|r|r|r|r||}
\hline
n&n_1&n_2&\text{Link}_{\text{largest}}&\text{Link}_{\text{MRD}}\ &\text{Lifted MRD}& \text{Extended Lifted MRD}\\\hline\hline
12&6&6   &262,221&262,208&   262,144 &   266,304\\ \hline
13&6&7   &1,048,905&1,048,832& 1,048,576 & 1,065,216\\ \hline
14&6&8   &4,195,616&4,195,328& 4,194,304 & 4,260,864  \\ \hline
\end{array}
\]
We will return to these codes in Theorem~\ref{T-LinkMRD} when we investigate decoding.
\end{exa}
%%%%%%%%%%%%%%%%%%%%%%%%%%%%%

The linkage construction can be viewed as a generalization of two specific constructions that can be found in the literature.
We will discuss the details in the next section, where we turn to subspace codes with largest possible distance.

%%%%%%%%%%%%%%%%%%%%%%%%%%%%%%%%%%%%%%
\section{Partial Spreads}\label{S-ParSpr}
%%%%%%%%%%%%%%%%%%%%%%%%%%%%%%%%%%%%%%
With the aid of Theorem~\ref{T-Linkage} we can construct optimal partial spreads for certain cases.
Recall that a \emph{partial spread} in~$\F^n$ is a collection of subspaces that pairwise intersect trivially.
If all subspaces have the same dimension, say~$k$, then this is simply a constant-dimension code of dimension~$k$
and distance~$2k$, and we call the code a \emph{partial $k$-spread}.
It is well known that if~$k$ divides~$n$, then an optimal partial $k$-spread (i.e., a partial $k$-spread of maximum cardinality)
is a $k$-\emph{spread}, i.e., the spaces intersect trivially and cover the entire~$\F^n$.
In this case a simple counting argument shows that the cardinality is~$(q^n-1)/(q^k-1)$, where $\F=\F_q$.
Several constructions of $k$-spreads are known.
For later reference we provide the following two options.
%%%%%%%%%%%%%%%%%%%%%%%%%%%
\begin{rem}\label{R-Spread}
Let $k$ divide~$n$.
\begin{alphalist}
\item \cite[Thm.~11]{RoTr13}
      The orbit of the subfield~$\F_{q^k}$ in the field~$\F_{q^n}$ under the natural action of the group~$\Fqns$ is a
      $k$-spread in~$\F_{q^n}$.
\item \cite[Thm.~6, Rem.~8]{GMR12} Let $m=n/k$ and $M\in\GL_{k}(\F)$ be the companion matrix of a primitive polynomial of degree~$k$.
      Then the set $\{\im(A_1,\ldots,A_m)\mid A_i\in\F_q[M],\text{ not all~$A_i$ are zero}\}$ is a $k$-spread in~$\F^n$.
      It is called a \emph{Desarguesian spread}.
\end{alphalist}
\end{rem}
%%%%%%%%%%%%%%%%%%%%%%%%%
If~$k$ does not divide~$n$, then the maximum size of a partial $k$-spread in~$\F_q^n$ is in general not known -- with one
exception which will be considered below in further detail.
The following result can be found in~\cite[Thms.~4.1, 4.2]{Beu75} and \cite[Thm.~7]{DrFr79};
see also \cite[Thm.~3]{EJSSS10}.
%%%%%%%%%%%%%%%%%%%%%%
\begin{theo}\label{T-MaxPartSpread}
Let $n\:(\mod\:k)=c$, where $c\in\{0,\ldots,k-1\}$.
Denote the largest possible cardinality of a partial $k$-spread in $\F_q^n$ by~$\mu(n,k)$. Then
\[
   \mu(n,k)\geq\frac{q^n-q^c}{q^k-1}-q^c+1
\]
with equality if $c\in\{0,1\}$.
Furthermore, if $c>1$ then
\[
   \mu(n,k)\leq \frac{q^n-q^c}{q^k-1}-\lfloor\theta\rfloor-1, \text{ where }
   \theta=\frac{\sqrt{1+4q^k(q^k-q^c)}-(2q^k-2q^c+1)}{2}.
\]
\end{theo}
%%%%%%%%%%%%%%%%%%%%%%%%
Constructions of partial $k$-spreads and cardinality
\begin{equation}\label{e-cardmin}
  m(n,k):=\frac{q^n-q^c}{q^k-1}-q^c+1, \text{ where } n\:(\mod\:k)=c,
\end{equation}
were presented in \cite[Thms.~4.2]{Beu75} as well as \cite[Thm.~11]{EtVa11} and \cite[Thm.~13]{GoRa14}.
Hence for~$c\in\{0,1\}$ these partial spreads have maximum possible cardinality.
The latter two constructions are special cases of our linkage and will be described in our terminology in the following examples.

%%%%%%%%%%%%%%%%%%%%%%%%%%%%%
\begin{exa}\label{E-EtVa11}
We describe the construction of partial $k$-spreads by Etzion/Vardy~\cite[Thm.~11]{EtVa11}.
Let $n\geq2k$ and write $n=lk+c$, where $c\in\{0,\ldots,k-1\}$.
Set $n_1=k(l-1)$ and $n_2=k+c$.
Consider the $\F_q$-vector space $\F_{q^{n_1}}\times\F_{q^{n_2}}$.
In~$\F_{q^{n_1}}$ choose the $k$-spread~$\cC_1$ given by the orbit of the subfield~$\F_{q^k}$ under the action of the
cyclic group $\F_{q^{n_1}}^*$; see Remark~\ref{R-Spread}(a).
Furthermore, let~$\beta$ be a primitive element of~$\F_{q^{n_2}}$ and set
$\cC_2=\big\{\spann_{\F}\{1,\,\beta,\ldots,\,\beta^{k-1}\}\big\}$.
Note that this is trivially a partial $k$-spread in~$\F_{q^{n_2}}$ of maximal possible cardinality because $k>n_2/2$.
Consider the coordinate map w.r.t.\ the basis $\{1,\,\beta,\ldots,\,\beta^{n_2-1}\}$ of~$\F_{q^{n_2}}$, that is,
\[
  \varphi:\F_{q^{n_2}}\longrightarrow\F_q^{n_2},\quad \sum_{i=0}^{n_2-1} f_i\beta^i\longmapsto(f_0,\ldots,f_{n_2-1}).
\]
Using the identification~$\varphi$, the code~$\cC_2$ simply translates into $\cC_2=\{\im(I_k\mid 0_{k\times c})\}$.
Finally, in~$\F^{k\times n_2}$ choose the rank-metric code
\[
    \cC_R=\{(I_k\mid 0_{k\times c})M^j\mid j=0,\ldots,q^{n_2}-2\}\cup\{0\},
\]
where~$M$ is the companion matrix of the minimal polynomial of~$\beta$ over~$\F_q$.
Note that the matrix $(I_k\mid 0_{k\times c})M^j$ consist exactly of the rows $\varphi(\beta^j),\ldots,\varphi(\beta^{j+k-1})$.
By Remark~\ref{R-COC}, $\cC_R$ is an MRD code of rank distance~$k$.
Identifying $\F_{q^{n_1}}$ with~$\F_q^{n_1}$, the linkage code $\cC_1\ast_{\cC_R}\cC_2$
is exactly the partial $k$-spread constructed in \cite[Thm.~11]{EtVa11}.
It has cardinality $1+q^{n_2}(q^{n_1}-1)/(q^k-1)$, and this is $m(n,k)$.
\end{exa}
%%%%%%%%%%%%%%%%%%%%%%%%%%%%%%%%%%

%%%%%%%%%%%%%%%%%%%%%%%%%%%%%%%%%%
\begin{exa}\label{E-GoRa14}
Essentially the same construction as in Example~\ref{E-EtVa11} but with different specifications of the constituent codes
is used by Gorla/Ravagnani in~\cite[Thm.~13]{GoRa14}.
Again, let $n=lk+c$ and set $n_1=k(l-1)$ and $n_2=k+c$.
Then the code constructed in~\cite[Thm.~13]{GoRa14} is the linkage $\cC_1\ast_{\cC_R}\cC_2$ with the following specifications:
$\cC_1$ is a Desarguesian $k$-spread in~$\F^{n_1}$ (see~Remark~\ref{R-Spread}) while
$\cC_2$ is the subspace code $\{\im(0_{k\times c}\mid I_k)\}$ and~$\cC_R$ is an
MRD code in~$\F^{k\times n_2}$ as in Remark~\ref{R-COC} with matrix~$W=(0\mid I_k)$, thus the nonzero
matrices in~$\cC_R$ are the last~$k$ rows of the matrices~$M^l$.
In addition to the construction, the authors also present a decoding algorithm for their partial spreads by making explicit use of
the structure of the Desarguesian spread; see~\cite[Sec.~5]{GoRa14}.
In contrast, no decoding algorithm is given in~\cite{EtVa11} for the partial spreads constructed therein.
\end{exa}
%%%%%%%%%%%%%%%%%%%%%%%%%%%%%%%%

Instead of partitioning~$n$ into $n_1+n_2$ with the specific choice of~$n_2=k+c$ as in the previous examples,
we may use any other splitting $n=lk+n_2$.
This will be summarized in the next result where we also address maximality of the partial spread.
A partial $k$-spread in~$\F^n$ is called \emph{maximal} if it is maximal with respect to inclusion, that is, it is not properly
contained in any other partial $k$-spread.
The following result shows, among other things, that linking a $k$-spread and a maximal partial $k$-spread through an MRD code
leads to a maximal partial $k$-spread.
%%%%%%%%%%%%%%%%%%%%%%%%%%%%%%%%%
\begin{theo}\label{T-MaxParSpread}
Let $n=lk+n_2$, where $l\geq 1$ and $n_2\geq k$.
Let~$\cC_1$ be a $k$-spread in~$\F^{lk}$ and $\cC_2$ be a partial $k$-spread in~$\F^{n_2}$.
Furthermore, let~$\cC_R$ be a linear MRD code in~$\F^{k\times n_2}$ with rank distance~$k$ and thus cardinality $q^{n_2}$.
Finally, let~$\cC=\cC_1\ast_{\cC_R}\cC_2$ be the resulting linkage code as in Theorem~\ref{T-Linkage}.
\begin{alphalist}
\item If~$|\cC_2|=m(n_2,k)$, then $|\cC|=m(n,k)$.
\item If~$\cC_2$ is a maximal partial $k$-spread then so is~$\cC$.
\end{alphalist}
\end{theo}
%%%%%%%%%%%%%%%%%%%%%%%%%%%%%%%%
\begin{proof}
(a) Theorem~\ref{T-Linkage} tells us that $|\cC|=q^{n_2}m(lk,k)+m(n_2,k)$. But this is easily seen to be~$m(n,k)$.
\\
(b) Let $(W_1\mid W_2)\in\F^{k\times(n_1+n_2)}$ be of rank~$k$ and set $\cW:=\im(W_1\mid W_2)$.
We have to show that there exists a subspace~$\cV\in\cC$ such that $\cW\cap\cV\neq\{0\}$, for then~$\cC$ is a maximal partial $k$-spread.
Assume first that~$W_1\neq0$.
Then there exists $(x,\,y)\in\cW$ such that $x\neq0$.
Since~$\cC_1$ is a spread of~$\F^{lk}$, the vector $x$ is in exactly one subspace of~$\cC_1$, say
$\im(U_1)$.
Let $x=\alpha U_1$, where $\alpha\in\F^k\backslash\{0\}$.
Since~$\cC_R$ is a linear rank-metric code with rank distance~$k$, we have $\alpha M\neq\alpha M'$ for all distinct $M,\,M'\in\cC_R$.
This shows that the set $\{\alpha M\mid M\in\cC_R\}$ has cardinality $|\cC_R|=q^{n_2}$ and therefore equals $\F^{n_2}$.
As a consequence, $y=\alpha M$ for some $M\in\cC_R$.
Hence $(x,y)=\alpha(U_1\mid M)$ and $\cW\cap\cV\neq\{0\}$ for the subspace $\cV=\im (U_1\mid M)\in\cC$.
Assume now $\cW=\im(0\mid W_2)$. Then $\rank(W_2)=k$ and the maximality of~$\cC_2$ implies $\cW\cap\cV\neq\{0\}$ for some
subspace $\cV=\im(0\mid U_2)\in\cC$.
All of this shows that~$\cC$ is a maximal partial $k$-spread in~$\F^n$.
\end{proof}

The following is an immediate consequence of Examples~\ref{E-EtVa11} and~\ref{E-GoRa14} because in both cases the chosen code~$\cC_2$
is trivially a maximal partial $k$-spread.
%%%%%%%%%%%%%%%%%%%%%%%%%%%%%%
\begin{cor}\label{C-maximalPS}
The partial spreads constructed in~\cite[Thm.~11]{EtVa11} and in \cite[Thm.~13]{GoRa14}
are maximal.
\end{cor}
%%%%%%%%%%%%%%%%%%%%%%%%%%%%%%%
Maximality of the partial spreads in~\cite{GoRa14} has also been established by Gorla/Ravagnani in~\cite[Prop.~20]{GoRa14}.

In Theorem~\ref{T-MaxPartSpread} we have seen that the maximum cardinality of a partial
$k$-spread in~$\F^n$ is known whenever $n\:(\mod\:k)\in\{0,1\}$.
There is one more case where the cardinality is known, and that is if $q=2$ and~$k=3$.
The following result covers all remainders of~$n$ modulo~$3$.

%%%%%%%%%%%%%%%%%%%%%%%%%%
\begin{theo}[\mbox{\cite[Thm.~5]{EJSSS10}}]\label{T-P3SpreadOpt}
Let $k=3$ and $n\geq6$. Let $n\:(\mod\:3) =c$. Then the maximum cardinality of a partial 3-spread in~$\F_2^n$
is
\[
  \frac{2^n-2^c}{7}-c.
\]
We call a partial~$3$-spread with this cardinality a maximum partial $3$-spread.
\end{theo}
%%%%%%%%%%%%%%%%%%%%%%%%
Note that for $c\in\{0,1\}$ the result is simply a special case of Theorem~\ref{T-MaxParSpread}, whereas for~$c=2$
the cardinality $m(n,k)$ in~\eqref{e-cardmin} is one below the maximum.
As a consequence, the constructions in~\cite[Thm.~11]{EtVa11} and~\cite[Thm.~13]{GoRa14} are just one
subspace short of being maximum.

The proof of Theorem~\ref{T-P3SpreadOpt} is based on a concrete example for $n=8$ and an extension construction for $n>8$.
It makes use of a result in~\cite[Lem.~4]{Bu80}, which establishes a partition of~$\F_q^n$ into subspaces
of two distinct dimensions.
Below we will provide an alternative extension, where we will also make use of the maximum $3$-spread in~$\F_2^8$.
%%%%%%%%%%%%%%%%%%%%%%%
\begin{exa}[\mbox{\cite[Ex.~2]{EJSSS10}}]\label{E-P3Spread}
There exists a partial 3-spread in~$\F_2^8$ with cardinality~$34$. Hence the spread is maximum.
It has been found by computer search and is explicitly given in~\cite{EJSSS10}.
\end{exa}
%%%%%%%%%%%%%%%%%%%%%%%

Now we can provide a simple construction of maximum partial $3$-spreads in~$\F_2^n$ for any~$n\geq10$.
Note that, due to the previous example and earlier discussions,
a maximum partial $3$-spread in~$\F_2^n$ is available for the values $n\in\{6,7,8,9\}$.

%%%%%%%%%%%%%%%%%%%%%%%%%%%%%%%
\begin{cor}\label{C-Opt3Spread}
Let $n\geq 10$ and write $n=3l+n_2$ for some $l\geq1$ and $n_2\in\{6,\,7,\,8\}$.
Choose a 3-spread~$\cC_1$ in $\F_2^{3l}$ and a maximum partial
3-spread~$\cC_2$ in~$\F_2^{n_2}$.
Finally, let~$\cC_R$ be an MRD code with rank distance~$3$ in~$\F_2^{3\times n_2}$.
Then $\cC_1\ast_{\cC_R}\cC_2$ is a maximum partial $3$-spread in~$\F_2^n$.
\end{cor}
%%%%%%%%%%%%%%%%%%%%%%%%%%%%%%
\begin{proof}
The resulting code is certainly a partial spread.
Let~$n\:(\mod\:3)=c$, thus $n_2\:(\mod\:3)=c$.
By Theorems~\ref{T-P3SpreadOpt} and~\ref{T-Linkage}  the cardinality of $\cC_1\ast_{\cC_R}\cC_2$ is given by
\[
  2^{n_2}\frac{2^{3l}-1}{7}+\frac{2^{n_2}-2^c}{7}-c=\frac{2^n-2^c}{7}-c,
\]
and this is the maximum value due to Theorem~\ref{T-P3SpreadOpt}.
\end{proof}

%%%%%%%%%%%%%%%%%%%%%%%%%%%%%%%%%%%%%%%%
\section{Decoding of Linkage Codes}\label{S-Decoding}
In this section we turn to decodability of the linkage codes from Theorem~\ref{T-Linkage}.
Of course, one aims at reducing decoding of $\cC_1\ast_{\cC_R}\cC_2$ to decoding of the smaller codes
$\cC_1,\,\cC_2,\,\cC_R$.
We will show first that this strategy does not work if one utilizes the rank metric for the code~$\cC_R$.
Instead one has to employ the subspace distance for all codes involved.
We will show that if we use suitable MRD codes and liftings thereof, then decoding can indeed be reduced to
decoding of the constituent codes.
Since lifted Gabidulin codes can be efficiently decoded, as proven by Silva et al.~\cite{SKK08}, this leads
to an efficient decoding algorithm for a particular instance of linkage codes.

The following terminology is standard.

%%%%%%%%%%%%%%%%%%%%%%%%%%
\begin{defi}\label{D-decodable}
Let~$\cC$ be a subspace code in~$\F^n$ with subspace distance~$d$.
A subspace~$\cV\subseteq\F^n$ is called \emph{decodable} w.r.t.~$\cC$ if there exists a subspace $\cU\in\cC$ such that
$\ds(\cU,\,\cV)\leq\frac{d-1}{2}$.
\end{defi}
%%%%%%%%%%%%%%%%%%%%%%%%%%%%
Since the subspace distance is a metric on the set of all subspaces in~$\F^n$ (see~\cite[Lem.~1]{KoKsch08}), a decodable subspace has a unique
closest codeword in~$\cC$.

The following simple fact will be useful later.
%%%%%%%%%%%%%%%%%%%%%
\begin{rem}\label{R-UcapV}
Let~$\cC$ be a constant-dimension code in~$\F^n$ with dimension~$k$ and subspace distance~$d$.
Suppose~$\cV\subseteq\F^n$ is a decodable $K$-dimensional subspace and~$\cU\in\cC$ is the unique closest codeword.
Then
\[
    \dim(\cU\cap\cV)>K/2\ \text{ and }\ \dim(\cU+\cV)<k+K/2.
\]
The first inequality follows from $\ds(\cU,\,\cV)\leq(d-1)/2<k$, see~\eqref{e-dmax}, which then reads as
$K+k-2\dim(\cU\cap\cV)<k$.
The second inequality is obtained by using $\dim(\cU+\cV)=k+K-\dim(\cU\cap\cV)$.
\end{rem}
%%%%%%%%%%%%%%%%%%%%%%

We start with an example illustrating that the rank distance of the code~$\cC_R$ cannot be used in the natural way for decoding
the linkage code $\cC_1\ast_{\cC_R}\cC_2$.

Throughout this section we call a matrix~$V\in\F^{K\times n}$ a \emph{matrix representation} of the subspace~$\cV\subseteq\F^n$ if
$\cV=\im(M):=\{xM\mid x\in\F^K\}$. We explicitly allow $\dim(\cV)<K$.

%%%%%%%%%%%%%%%%%%%%%%%%%%%%
\begin{exa}\label{E-BadDec}
Let $n_1=n_2=8,\,k=4$ and $\F=\F_2$. Define
\[
  \cC_1=\cC_2=\{\im(I\mid I),\,\im(I\mid 0)\}.
\]
where $I$ and~$0$ are the identity and the zero matrix in~$\F^{4\times4}$, respectively.
Moreover, let
\[
  \cC_R=\big\{0_{4\times 8},\, (I\mid0),\,(M\mid 0),\,(I+M\mid 0)\big\}, \text{ where } M=\begin{pmatrix}0&1&0&0\\0&0&1&0\\0&0&0&1\\1&1&0&0\end{pmatrix}.
\]
Then $\ds(\cC_1)=\ds(\cC_2)=8$ and $d_R:=\dr(\cC_R)=4$.
Thus the linkage code~$\cC=\cC_1\ast_{\cC_R}\cC_2$ has length~$16$, subspace distance~$d=8$ and cardinality $N=10$.
Consider the received word
\[
  \cV=\im\left(\begin{array}{cccccccc|cccccccc}1&1&0&0&1&1&0&0&0&1&1&0&0&0&0&0\\0&1&1&0&0&1&1&0&0&0&1&1&0&0&0&0\end{array}\right).
\]
Then $\ds(\cV,\,\cU)=2$ for the codeword $\cU=\im(U_1\mid U_2)\in\cC$, where
\[
  U_1=(I\mid I)\text{ and }U_2=(M\mid 0).
\]
In particular,~$\cV$ is decodable.
Note that $\cV\subseteq\cU$ (i.e., only erasures occurred during the transmission).
One can check straightforwardly that there exists no matrix representation $(V_1\mid V_2)$  of~$\cV$ in $\F^{4\times16}$
such that $\dr(V_2,\,U_2)\leq \frac{d_R-1}{2}=\frac{3}{2}$.
Even worse, for all matrix representations $(V_1\mid V_2)\in\F^{4\times16}$ of~$\cV$ for which the matrix~$V_2$ has a unique closest matrix in~$\cC_R$ with respect
to the rank distance, this unique closest matrix is the zero matrix and therefore does not lead to the correct decoding~$\cU$.
The fact that the ``obvious decoding'' does not work may be explained by the fact that the subspaces represented by the
nonzero matrices in~$\cC_R$, i.e., $\im(I\mid0),\,\im(M\mid 0),\,\im(I+M\mid 0)$, all coincide.
This causes the rank-metric code~$\cC_R$ to be of little help with decoding.
\end{exa}
%%%%%%%%%%%%%%%%%%%%%%%%%%%%

The example can be generalized.
We introduce the following notation.
Define the projections
\begin{equation}\label{e-proj}
  \pi_1:\F^{n_1+n_2}\longrightarrow \F^{n_1},\ (a,b)\longmapsto a \quad \text{and}\quad
  \pi_2:\F^{n_1+n_2}\longrightarrow \F^{n_2},\ (a,b)\longmapsto b.
\end{equation}
For a subspace $\cA\subseteq\F^n$ we define $\cA_i=\pi_i(\cA)$.
Thus, if $\cA=\im(A_1\mmid A_2)$, then $\cA_i=\im(A_i)$ for $i=1,2$.

%%%%%%%%%%%%%%%%%%%%%%%%%%%%%%%
\begin{prop}\label{P-BadDec}
Let~$\cC$ be as in Theorem~\ref{T-Linkage} and assume $d\geq d_R+2$.
Then there exists a subspace $\cU=\im(U_1\mmid U_2)\in\cC$ and a received word~$\cV\subseteq\F^n$ such that
\begin{alphalist}
\item $\ds(\cU,\cV)\leq\frac{d-1}{2}$ (that is,~$\cV$ is decodable),
\item $\cV\subseteq\cU$ (hence only erasures occurred during transmission),
\item for any $V_2\in\F^{k\times n_2}$ such that $\im(V_2)=\pi_2(\cV)$ we have
      \[
          \rank(V_2-U_2)>\frac{d_R-1}{2}.
      \]
      In other words, it is not possible to decode~$\cV$ by making use of the rank metric for the code~$\cC_R$.
\end{alphalist}
\end{prop}
%%%%%%%%%%%%%%%%%%%%%%%%%%%%%%
\begin{proof}
Let $d\geq d_R+2$. Since $d\leq d_1\leq 2k$, we have $r:=k-\lceil\frac{d_R}{2}\rceil\geq\frac{1}{2}$, and
thus $r\geq1$ because it is an integer.
We construct now subspaces~$\cU$ and~$\cV$ as stated in the proposition.
First choose a subspace $\cU=\im(U_1\mmid U_2)\in\cC$ with $(U_1\mmid U_2)\in\F^{k\times n}$ such that
$\rank(U_2)\geq k-r=\lceil\frac{d_R}{2}\rceil$.
By definition of~$\cC$ such an element does indeed exist.
Next, there exists a matrix $X\in\GL_k(\F)$ such that
\[
    X(U_1\mmid U_2)=\begin{pmatrix}U_{11}\!\!\!&\mid\, U_{21}\\ U_{12}\!\!\!&\mid\, U_{22}\end{pmatrix}
\]
and where $U_{22}\in\F^{(k-r)\times n_2}$ has rank~$k-r$ and and $\im(U_{21})\cap\im(U_{22})=\{0\}$.
Put
\[
   \cV=\im(U_{11}\mmid U_{21}).
\]
Then $\dim(\cV)=r$ because the rows of~$(U_1\mmid U_2)$ are linearly independent.
Moreover,~$\cV\subseteq\cU$ and therefore
\[
  \ds(\cV,\cU)=k+r-2r=\big\lceil\frac{d_R}{2}\big\rceil\leq \frac{d-1}{2}.
\]
This establishes~(a) and~(b).
For~(c) consider now all matrices~$V_2$ in $\F^{k\times n_2}$ whose row space is $\pi_2(\cV)$.
These matrices can be written as

\[
    V_2:=X^{-1}\begin{pmatrix}M_1\\M_2\end{pmatrix}U_{21},\text{ where }
         \begin{pmatrix}M_1\\M_2\end{pmatrix}\in\F^{k\times r} \text{ is any matrix of rank }r.
\]
The matrix~$X^{-1}$ does not change the row space, and we include it only to simplify the next step.
Indeed, for each such matrix~$V_2$ we have
\[
   \rank(V_2-U_2)=\rank(XV_2-XU_2)=\rank\begin{pmatrix}M_1U_{21}-U_{21}\\M_2U_{21}-U_{22}\end{pmatrix}
        \geq \rank(M_2U_{21}-U_{22}).
\]
The rightmost matrix has full row rank,~$k-r$.
Indeed, suppose $u(M_2U_{21}-U_{22})=0$.
Then $uM_2U_{21}=uU_{22}\in\im(U_{21})\cap\im(U_{22})$.
Since this intersection is trivial, we obtain $uU_{22}=0$, which in turn implies $u=0$.
All of this shows that $\rank(U_2-V_2)\geq \lceil\frac{d_R}{2}\rceil>\frac{d_R-1}{2}$ for all matrix representations of~$\pi_2(\cV)$.
\end{proof}

The last observation suggests to modify the linkage construction by simply replacing the rank-metric
code~$\cC_R$ by matrix representations of a subspace code.
This results in a code that is decodable if its constituent codes are decodable.
But since these codes are considerably smaller than the original linkage codes, we will not follow that path.

Instead, we will show now how to decode linkage codes $\cC_1\ast_{\cC_R}\cC_2$ for the case where~$\cC_1$ and~$\cC_R$ are (lifted) MRD codes.
We need the following lemma.

%%%%%%%%%%%%%%%%%%%%%%
\begin{lemma}\label{L-V1U1}
Let~$\cC$ be the code from Theorem~\ref{T-Linkage} and $\cV\subseteq\F^n$ be a decodable $K$-dimensional subspace.
Let $\cU\in\cC$ be the closest codeword, thus $\ds(\cU,\,\cV)\leq \frac{d-1}{2}$.
Then
\[
   \cU_1=0\Longleftrightarrow \dim\cV_1\leq K/2.
\]
\end{lemma}
%%%%%%%%%%%%%%%%%%%%%%%%
\begin{proof}
``$\Longleftarrow$''
Assume $\cU_1\neq 0$. Then $\rank(U_1)=k$ by definition of~$\cC$.
Thus $\dim(\cU)=k=\dim(\cU_1)$, and the map $\pi_1|_{\cU}$ is injective, where~$\pi_1$ is the projection from~\eqref{e-proj}.
We compute
\[
  \dim(\cU\cap\cV)=\dim(\pi_1(\cU\cap\cV))\leq\dim(\pi_1(\cU)\cap\pi_1(\cV))\leq\dim(\pi_1(\cV))=\dim\cV_1\leq K/2,
\]
which is a contradiction to Remark~\ref{R-UcapV}.
Thus $\cU_1=0$.
\\
``$\Longrightarrow$''
Let $\cU_1=0$, thus $\cU=\im(0\mid U_2)$ and $\rank(U_2)=k$ by definition of~$\cC$.
Write $\cV=\im(V_1\mmid V_2)$ for some $(V_1\mmid V_2)\in\F^{K\times n}$.
With the aid of Remark~\ref{R-UcapV} we obtain
\[
  \rank(V_1)+k\leq\rank\begin{pmatrix}0&U_2\\V_1&V_2\end{pmatrix}=\dim(\cU+\cV)<k+\frac{K}{2}.
\]
Hence $\dim\cV_1=\rank(V_1)<K/2$.
\end{proof}
Note that the implication ``$\Longleftarrow$'' of the last lemma is in general not true for~$\cU_2$ and~$\cV_2$
because the matrices in~$\cC_R$ may not have rank~$k$.

Now we are in the position to discuss decoding of the linkage codes from Theorem~\ref{T-Linkage}.
We consider the following situation which is a special case of the general linkage construction.
It may also be regarded as an extension of the codes considered in Example~\ref{E-MRDLinkage}.
%%%%%%%%%%%%%%%%%%
\begin{theo}\label{T-LinkMRD}
For $i=1,\,2$ let $n_i\geq k$  and let $\cM_i\subseteq\F^{k\times n_i}$ be a linear MRD code with rank distance~$d$, thus $|\cM_i|=q^{n_i(k-d+1)}$.
Moreover, let $\cM_3\subseteq\F^{k\times n_1}$ and $\cM_4\subseteq\F^{k\times n_2}$ be  SC-representing sets of constant-dimension codes with subspace distance~$2d$.
Consider the code $\cC=\cC'\cup\cC''\cup\cC'''$, where
\begin{align*}
  \cC'&=\{\im(I_k\mmid M_1\mmid M_2)\mid M_1\in\cM_1,\, M_2\in\cM_2\},\\
  \cC''&=\{\im(0_{k\times k}\mmid M\mmid 0_{k\times n_2})\mid M\in\cM_3\},\\
  \cC'''&=\{\im(0_{k\times k}\mmid  0_{k\times n_1}\mmid M)\mid M\in\cM_4\}.
\end{align*}
Then $\cC$ is an $(n,\,N,\,k,\,2d)$-code, where $n=k+n_1+n_2$ and $N=q^{(n_1+n_2)(k-d+1)}+|\cM_3|+|\cM_4|$.
\end{theo}
%%%%%%%%%%%%%%%%%%%
\begin{proof}
This is a simple application of Theorem~\ref{T-Linkage} : $\cC=\tilde{\cC}_1\ast_{\tilde{\cC}_R}\tilde{\cC}_2$ with the codes
$\tilde{\cC}_1=\{\im(I_k)\}$, $\tilde{\cC}_2=\{\im(M\mmid 0)\mid M\in\cM_3\}\cup\{\im(0\mid M)\mmid M\in\cM_4\}$, and
$\tilde{\cC}_R=\{(M_1\mid M_2)\mid M_i\in\cM_i\}$.
\end{proof}

Note that if~$\cM_4$ represents a lifted MRD code, then $\cC'\cup\cC'''=\cC_1\ast_{\cM_2}\cC_2$, where
$\cC_1=\{\im(I\mmid M_1)\mid M_1\in\cM_1\},\,\cC_2=\{\im(M)\mid M\in\cM_4\}$.
Hence the code is of the form as discussed in Example~\ref{E-MRDLinkage}, and in Theorem~\ref{T-LinkMRD} we improve upon the codes
in that example by the size of~$\cM_3$.
In the column ``Link$_{\text{MRD}}$'' of the table in Example~\ref{E-MRDLinkage} we listed, for a specific choice of
parameters, the largest codes of the form $\cC'\cup\cC'''$ above.
As we saw already, the largest size is attained when~$n_2$ is largest subject to $n=k+n_1+n_2$ with $n_1+k\geq 2k$ and $n_2\geq 2k$
(now $n_1+k$ takes the role of~$n_1$ from that example), thus for $n_1=k$ and $n_2=n-2k$.
But in that case~$|\cM_3|=1$ and thus we improve upon the codes in that table by exactly one subspace.

We now turn to decoding of the codes in Theorem~\ref{T-LinkMRD}.
As we show next this can be reduced to decoding of the constituent codes.

%%%%%%%%%%%%%%%%%%
\begin{theo}\label{T-DecLinkGab}
Consider the setting of Theorem~\ref{T-LinkMRD}.
To ease notation we set $n_0:=k$.
For $i=1,2$ define the lifted MRD codes $\cC_i:=\{\im(I_k\mmid M)\mid M\in\cM_i\}$.
Furthermore, set~$\cC_3:=\cC(\cM_3)$ and $\cC_4:=\cC(\cM_4)$.
Then, if~$\cC_1,\ldots,\,\cC_4$ are decodable then so is the linkage code~$\cC$.
More precisely, let~$\cV=\im(V_0\mmid V_1\mmid V_2)\subseteq\F^n,\, V_i\in\F^{K\times n_i}$, be a $K$-dimensional received word such that
$\ds(\cV,\cC)\leq (2d-1)/2$. Then exactly one of the following situations occurs.
\begin{alphalist}
\item $\rank(V_0\mmid V_1)< K/2$.
      In this case the unique closest codeword in~$\cC$  is in~$\cC'''$ and given by $\cU=\im(0\mmid 0\mmid M)$, where $M\in\cM_4$ is the unique matrix
      such that $\ds(\im(M),\,\im(V_2))\leq(2d-1)/2$.
\item $\rank(V_0\mmid V_2)< K/2$.
      In this case the unique closest codeword in~$\cC$  is in~$\cC''$ and given by $\cU=\im(0\mmid M\mmid 0)$, where $M\in\cM_3$ is the unique matrix
      such that $\ds(\im(M),\,\im(V_1))\leq(2d-1)/2$.
\item $\rank(V_0)>K/2$. In this case the unique closest codeword in~$\cC$  is in~$\cC'$ and given by $\cU=\im(I\mmid M_1\mmid M_2)$, where
      $M_i \in\cM_i$ are the unique matrices such that
      $\ds\big(\im(I\mmid M_i),\,\im(V_0\mmid V_i)\big)\leq(2d-1)/2$ for $i=1,2$.
\end{alphalist}
\end{theo}
%%%%%%%%%%%%%%%%%%

\begin{proof}
First of all, the uniqueness of the matrices~$M$ and~$M_i$ in (a) --~(c) is guaranteed since the subspace codes
$\cC(\cM_3),\,\cC(\cM_4)$ and the lifted MRD codes $\cC_i$ all have subspace distance~$2d$.

Let us denote  by $\cU=\im(U_0\mmid U_1\mmid U_2)$ the unique codeword in~$\cC$ closest to~$\cV$.
We will use the notation $\cU_i=\im(U_i)$ and $\cV_i=\im(V_i)$ for $i=0,1,2$.

First we show that at most one of the~3 cases above can occur. Clearly, if $\rank(V_0)>K/2$, then neither~(a) nor~(b) can occur.
Let now $\rank(V_0)\leq K/2$ and assume $r:=\rank(V_0\mmid V_1)< K/2$.
We have to show that $\rank(V_0\mmid V_2)\geq K/2$.
After suitable row operations we may assume
\[
    (V_0\mmid V_1\mmid V_2)=\begin{pmatrix}V_{01}&V_{11}&V_{21}\\ 0&0&V_{22}\end{pmatrix},
\]
where the first block row has~$r$ rows.
Then $\rank(V_0\mmid V_1\mmid V_2)=K$ implies that $\rank(V_{22})> K/2$.
This implies $\rank(V_0\mmid V_2)> K/2$, as desired.
Using symmetry, all of this shows that if~$\rank(V_0)\leq K/2$, then at most one of the cases~(a) or~(b) can occur.

Next we show that exactly one of the cases~(a) --~(c) occurs.
To do so, it suffices to show that if~$\rank(V_0)\leq K/2$ then~(a) or~(b) must occur.
We know that $\cC=\tilde{\cC}_1\ast_{\tilde{\cC}_R}\tilde{\cC}_2$
with~$\tilde{\cC}_1,\,\tilde{\cC}_2,\,\tilde{\cC}_R$ as in the proof of Theorem~\ref{T-LinkMRD}.
Therefore, Lemma~\ref{L-V1U1} along with $\rank(V_0)\leq K/2$ implies $U_0=0$.
Thus, $\cU\in\cC''\cup\cC'''$.
Let us assume $\cU\in\cC''$, say $\cU=\im(0\mmid M\mmid 0)$.
Along with $\rank(M)=k$ we derive
\begin{equation}\label{e-V0V2}
  \rank(V_0\mmid V_2)+k\leq\rank\begin{pmatrix}V_0&V_1&V_2\\0&M&0\end{pmatrix}=\dim(\cU+\cV)< k+K/2,
\end{equation}
where the last inequality is due to Remark~\ref{R-UcapV}.
As a consequence, $\rank(V_0\mmid V_2)<K/2$. Similarly,~$\cU\in\cC'''$ implies $\rank(V_0\mmid V_1)<K/2$.

Now we turn to decoding for each of the three cases.

(a) Let $\rank(V_0\mmid V_1)< K/2$.
Suppose the closest codeword~$\cU$ is in~$\cC''$, say $\cU=\im(0\mmid M\mmid 0)$.
Then~\eqref{e-V0V2} shows that $\rank(V_0\mmid V_2)<K/2$.
But this means that also case~(b) occurs, a contradiction.
Hence~$\cU\in\cC'\cup\cC'''$.
But this code is a linkage code. Indeed,
\[
  \cC'\cup\cC'''=\cC_1\ast_{\cC_R}\cC_4,
\]
where $\cC_1=\{\im(I_k\mmid M_1)\mid M_1\in\cM_1\}$, $\cC_4=\cC(\cM_4)$, and $\cC_R=\cM_2$.
Hence Lemma~\ref{L-V1U1} implies $(U_0\mmid U_1)=(0\mmid 0)$, thus $\cU=\im(0\mmid0\mmid M)\in\cC'''$ with $M\in\cM_4$.
In particular, $\rank(M)=k$.
Consider the projection~$\pi_2$ of~$\F^{k+n_1+n_2}$ onto~$\F^{n_2}$.
Then $(\pi_2)|_{\cU}$ is injective and thus $\dim(\cU\cap\cV)\leq\dim(\cU_2\cap\cV_2)$.
This implies
 $\ds(\cV_2,\,\cU_2)\leq K+k-2\dim(\cV_2\cap\cU_2)\leq \ds(\cV,\,\cU)\leq(2d-1)/2$.
Thus, decoding~$\cV_2$ to its closest codeword in~$\cC(\cM_4)$ results in~$\cU_2$.
Using its unique matrix representation $M\in\cM_4$, i.e., $\cU_2=\im(M)$, we arrive at the correct decoding
$\cU=\im(0\mmid0\mmid M)$ of the received space~$\cV$.

(b) The case $\rank(V_0\mmid V_2)< K/2$ is analogous.

(c)  Let $\rank(V_0)> K/2$.
Then Lemma~\ref{L-V1U1} applied to $\tilde{\cC}_1\ast_{\cC_R}\tilde{\cC}_2$ (see proof of Theorem~\ref{T-LinkMRD})
implies $(U_0)\neq 0$. Thus $\cU\in\cC'$.
In particular, we may assume $U_0=I_k$.
For $i=1,2$ let $\cV_i':=\im(V_0\mmid V_i)$ and $\cU_i':=\im(I_k\mmid U_i)$.
Then $\cU_i'\in\cC_i$ for $i=1,2$, where
$\cC_i=\{\im(I_k\mmid M)\mid M\in\cM_i\}$ is the lifting of the MRD code~$\cM_i$ for $i=1,2$.
Consider the projections
\[
  \psi_i:\F^{k+n_1+n_2}\longrightarrow \F^{k+n_i},\quad (a_0,a_1,a_2)\longmapsto (a_0,a_i)
\]
Then $(\psi_i)|_{\cU}$ is injective and thus $\dim(\cU\cap\cV)\leq\dim(\cU_i'\cap\cV_i')$.
As in~(a) this implies $\ds(\cU_i',\,\cV_i')\leq \ds(\cU,\,\cV)\leq(2d-1)/2$ for $i=1,2$.
Hence~$\cV_i'$ can be uniquely decoded w.r.t.\ $\cC_i$ and the closest codeword is given by~$\cU_i'$.
Using the unique matrix representations $(I\mmid U_i),\,U_i\in\cM_i$, of the spaces~$\cU_i$, we arrive at the correct decoding of~$\cV$.
\end{proof}

We summarize the result in the following algorithm.

%%%%%%%%%%%%%%%%%%%%%%%%%%%%%%%
\setcounter{algocf}{\value{theo}}
\begin{algorithm}\label{DecAlg}
 \KwData{a decodable $K$-dimensional subspace $\cV=\im(V_0\mmid V_1\mmid V_2)$ with
    $(V_0\mmid V_1 \mmid V_2)\in\F_q^{K\times n}$}
 \KwResult{the unique $\cU\in\cC'\cup\cC''\cup\cC'''$ such that $\dist{\cV,\cU}\leq\frac{2d-1}{2}$.}
	\eIf{$\rank(V_0\mmid V_1)<\frac{K}{2}$}{decode $\im(V_2)$ in $\cC(\cM_4)$ to $\im(U_{2})$\;
	\KwRet{$\cU=\im(0\mmid 0\mmid U_{2})$.}}{
	\eIf{$\rank(V_0\mmid V_2)<\frac{K}{2}$}{decode $\im(V_1)$ in $\cC(\cM_3)$ to $\im(U_{1})$\;
	\KwRet{$\cU=\im(0\mmid U_{1}\mmid 0)$.}}
	{decode $\im(V_0\mmid V_1)$ in $\cC_1$ to $\im(I_k\mmid U_{1})$\;
	 decode $\im(V_0\mmid V_2)$ in $\cC_2$ to $\im(I_k\mmid U_{2})$\;
	\KwRet{$\cU=\im(I_k\mmid U_{1}\mmid U_{2})$.}}}
 \caption{Decoding algorithm for the codes in Theorem~\ref{T-LinkMRD}}
\end{algorithm}
\addtocounter{theo}{1}
%%%%%%%%%%%%%%%%%%%%%%%%%%%%%%%%%%%%%%%

One should observe that in the last case of the algorithm, the two decoding steps can be performed in parallel.
A similar, but not identical, form of parallelizing decoding is also used for the spread codes in~\cite{GoRa14};
recall Example~\ref{E-GoRa14} for the relation to our linkage codes.

Clearly, the construction in Theorem~\ref{T-LinkMRD} and its decoding can easily be
generalized to more than~$3$ blocks.

%%%%%%%%%%%%%%%%%%
\begin{rem}\label{R-LinkGab}
A very efficient decoding is obtained when we use Gabidulin codes for~$\cM_1$ and~$\cM_2$ and lifted Gabidulin codes
for~$\cM_3$ and $\cM_4$  (thus~$n_i>k$). In this case, all codes relevant for decoding in the previous proof are
lifted Gabidulin codes, and the decoding algorithm derived by Silva et al.~\cite{SKK08} may be employed.
If $n_i>\!> k$, then even better efficiency is obtained by using direct products of Gabidulin codes as the MRD codes and the lifting
of such a code for~$\cM_3$ and~$\cM_4$; see~\cite[Sec.~VI.E]{SKK08}.
In fact, our code~$\cC'$ in Theorem~\ref{T-LinkMRD} (or rather its generalization to more than~$3$ blocks) is of the form
proposed in~\cite{SKK08} and with the above we have shown how to enlarge the code without compromising its properties.
In this sense, our results put the considerations in~\cite[Sec.~VI.E]{SKK08} in a broader context.
\end{rem}
%%%%%%%%%%%%%%%%%%%

%%%%%%%%%%%%%%%%%%%%%%%%%%%%%%%%%%%%%%%%%%%%%%%
\bibliographystyle{abbrv}
\bibliography{literatureAK,literatureLZ}
\end{document}